\documentclass[journal,onecolumn]{IEEEtran}

\usepackage{amsmath, amsthm, amssymb}
\usepackage{latexsym}
\usepackage{graphicx}
\usepackage{verbatim}
\usepackage{mathrsfs}
\usepackage{float}
\usepackage[lofdepth, lotdepth]{subfig}
\usepackage{array}
\usepackage{url}
\usepackage{algorithm}
\usepackage[noend]{algorithmic}
\usepackage{cite}
\usepackage[table]{xcolor}
\usepackage{enumerate}
\usepackage{eucal}
\usepackage{stmaryrd}
\usepackage{mathtools}
\usepackage{pgf}
\usepackage{tikz}
\usetikzlibrary{arrows,automata}
\usepackage[latin1]{inputenc}
\usetikzlibrary{automata,positioning}
\usepackage[margin=0.65in]{geometry}
\usepackage{mathtools}
\DeclarePairedDelimiter\ceil{\lceil}{\rceil}
\DeclarePairedDelimiter\floor{\lfloor}{\rfloor}

\theoremstyle{plain}
\newtheorem{theorem}{Theorem}
\newtheorem{proposition}[theorem]{Proposition}
\newtheorem{lemma}[theorem]{Lemma}

\newtheorem{corollary}[theorem]{Corollary}

\theoremstyle{definition}
\newtheorem{definition}[theorem]{Definition}
\newtheorem{example}[theorem]{Example}
\newtheorem{remark}[theorem]{Remark}

\newcommand{\B}{{\mathcal B}}
\newcommand{\C}{{\mathcal C}}
\newcommand{\D}{{\mathcal D}}

\DeclareMathAlphabet{\mathbfsl}{OT1}{ppl}{b}{it} 

\newcommand{\bL}{\mathbfsl{L}}

\newcommand{\bU}{\mathbfsl{U}}

\newcommand{\bp}{{\mathbfsl p}}

\newcommand{\by}{{\mathbfsl y}}

\newcommand{\bc}{{\mathbfsl c}}

\newcommand{\bx}{{\mathbfsl{x}}}
\newcommand{\bz}{{\mathbfsl{z}}}

\newcommand{\bsg}{{\boldsymbol{\sigma}}}

\newcommand{\bbZ}{{\mathbb Z}}

\newcommand{\ppmod}[1]{~({\rm mod~}#1)}

\renewcommand{\ge}{\geqslant}
\renewcommand{\le}{\leqslant}

\newcommand{\et}{{\emph{et al.}}}

\newcommand{\tA}{\mathtt{A}}
\newcommand{\tT}{\mathtt{T}}
\newcommand{\tC}{\mathtt{C}}
\newcommand{\tG}{\mathtt{G}}
\newcommand{\GC}{$\mathtt{GC}$}

\newcommand{\enc}{\textsc{Enc}}
\newcommand{\dec}{\textsc{Dec}}
\newcommand{\indexenc}{\textsc{IndexEnc}}

\newcommand{\Bindel}{{\cal B}^{\rm indel}}
\newcommand{\Bedit}{{\cal B}^{\rm edit}}

\begin{document}

\pagestyle{empty}

\title{Capacity-Approaching Constrained Codes with Error Correction for DNA-Based Data Storage}
\author{
   \IEEEauthorblockN{
   	Tuan Thanh Nguyen,
	Kui Cai, 
	Kees A. Schouhamer Immink,
	and Han Mao Kiah}
\thanks{Tuan Thanh Nguyen and Kui Cai are with the Singapore University of Technology and Design, Singapore 487372 (email: \{tuanthanh\_nguyen, cai\_kui\}@sutd.edu.sg).}
\thanks{Kees A. Schouhamer Immink is with the Turing Machines Inc, Willemskade 15d, 3016 DK Rotterdam, The Netherlands (email: immink@turing-machines.com).}
\thanks{Han Mao Kiah is with the School of Physical and Mathematical Sciences, Nanyang Technological University, Singapore 637371 (email: hmkiah@ntu.edu.sg).}
}

\maketitle

\hspace{-3mm}\begin{abstract}
We propose coding techniques that limit the length of homopolymers runs, ensure the ${\tt G}{\tt C}$-content constraint, and are capable of correcting a single edit error in strands of nucleotides in DNA-based data storage systems. In particular, for given $\ell, \epsilon>0$, we propose simple and efficient encoders/decoders that transform binary sequences into DNA base sequences (codewords), namely sequences of the symbols ${\tt A}, {\tt T}, {\tt C}$ and ${\tt G}$, that satisfy the following properties:
\begin{itemize}
\item Runlength constraint: the maximum homopolymer run in each codeword is at most $\ell$,
\item ${\tt G}{\tt C}$-content constraint: the ${\tt G}{\tt C}$-content of each codeword is within $[0.5-\epsilon,0.5+\epsilon]$,
\item Error-correction: each codeword is capable of correcting a single deletion, or single insertion, or single substitution error.
\end{itemize}
For practical values of $\ell$ and $\epsilon$, we show that our encoders achieve much higher rates than existing results in the literature and approach the capacity. 
Our methods have low encoding/decoding complexity and limited error propagation.
\end{abstract}

\section{Introduction}

In a DNA-based storage system, the input user data is translated into a large number of DNA strands (also known as DNA sequences or oligos), which are synthesized and stored in a DNA pool. To retrieve the original data, the stored DNA strands are sequenced and translated inversely back to the binary data. Several experiments have been conducted since 2012 (see \cite{Yazdi.2017, church2012, goldman2013, ross2013, fountain2017, exp1,Organick:2018}), and it has been found that substitutions, deletions, and insertions are common errors occurring at the stages of synthesis and sequencing. To improve the reliability of DNA storage, several channel coding techniques, including {\em constrained coding} and {\em error correction coding}, have been introduced \cite{immink2018,tt2019,wentu2018,dube2019,demarau}.

In a DNA strand, two properties that significantly increase the chance of errors for most synthesis and sequencing technologies are long {\em homopolymer run} \cite{ross2013, exp1} and high (or low) {\em GC-content}. A homopolymer run refers to the repetition of the same nucleotide. Ross \et{}\cite{ross2013} reported that a homopolymer run of length more than six would result in a significant increase of substitution and deletion errors (see \cite[Fig. 5]{ross2013}), and therefore, such long runs should be avoided. On the other hand, the \GC-content of a DNA strand refers to the percentage of nucleotides that are either $\tG$ or $\tC$, and DNA strands with \GC-content that are too high or too low are more prone to both synthesis and sequencing errors (see for example, \cite{Yakovchuk2006, ross2013}). Therefore, most experiments used DNA strands whose \GC-content is close to 50\% (for example, between 40\% to 60\% \cite{exp1}, or 45\% to 55\%\cite{fountain2017}).

Designing efficient constrained codes to translate binary data into DNA strands that satisfy the homopolymer runlength (also known as {\em runlength limited constraint, or RLL constraint in short}) and the \GC-content constraints has been a challenge.  In the literature, several prior art coding techniques have been introduced, mostly focusing on one specific value of maximum runlength or requiring \GC-content to be exactly 50\%, also known as {\em \GC-balanced constraint} \cite{immink2018, tt2019, wentu2018, dube2019}.
To encode \GC-balanced codewords, most works used a modification of the Knuth's balancing method for binary sequences \cite{knuth}. Since the constraint is strong, the coding redundancy is large (approximately $\log n$, where $n$ is the length of each codeword). In this work, we investigate the problem of translating binary data to DNA strands whose \GC-content is close to 50\%, and we refer this as {\em almost-balanced}. Via a simple modification of Knuth's method, we show that the number of redundant bits can be gracefully reduced from $\log n$ to $O(1)$.


Constrained codes can reduce the occurrence of substitution, deletion, and insertion errors in the DNA storage system. However, the constrained code itself cannot correct errors.  There are recent works that characterize the error probabilities by analyzing data from experiments and then demonstrate the need for error-correction codes. For example, Organick \et{} recently stored 200MB of data in 13 million DNA strands and reported substitution, deletion, and insertion rates to be $4.5\times 10^{-3}$, $1.5\times 10^{-3}$ and $5.4\times 10^{-4}$, respectively \cite{Organick:2018}. Since current technologies can only synthesize strands of DNA of one-two hundred nucleotides, it is most likely that there is at most one error of each type. Motivated by this error behavior, several works focused on the construction of error-correction codes that are capable of correcting the single edit (i.e. a single substitution, or a single deletion, or a single insertion) and its variants \cite{tt2019, demarau}. However, a problem of combining constrained codes with both the homopolymer runlength and \GC-content constraints with the single-edit-correction codes has not been addressed. 

In this work, we propose novel channel coding techniques for DNA storage, where the codebooks satisfy the RLL constraint, the \GC-content constraint, and can also correct a single edit and its variants. During the decoding of the proposed constrained codes, a small number of corrupted bits at the channel output might lead to massive error propagation of the decoded bits. Our proposed combination of constrained codes with error-correction codes also helps to minimize the error prorogation during decoding.

The paper is organized as follows. We first go through certain notations in Section II. In Section III, we present two efficient RLL coding methods that limit the maximum homopolymer run in each codeword to be at most $\ell$ for arbitrary $\ell>0$. Our methods are based on {\em enumeration coding} and {\em sequence replacement technique}, respectively. In Section IV, via a simple modification of Knuth's balancing method, we describe linear-time encoders/decoders that translate binary data to DNA strands whose GC-content is within $[0.5-\epsilon,0.5+\epsilon]$ for arbitrary $\epsilon>0$. This method yields a significant improvement in coding redundancy with respect to prior works. Then, in Section V, we present an efficient $(\ell,\epsilon)$-{constrained coding} method where codewords obey both RLL constraint and GC-content constraint. In Section VI, we modify the $(\ell,\epsilon)$-constrained coding so that the codewords can correct a single deletion, or single insertion, or single substitution error.

For the convenience of the reader, relevant notation and terminology referred to throughout the paper is summarized in Table~\ref{tab.notation}.

\begin{table}
\renewcommand{\arraystretch}{1.3}
\begin{tabular}{ p{5.8cm} p{11.8cm} }
 \hline
 Notation    & Description \\
 \hline
 $\Sigma$& alphabet of size $q$ \\
 $\Sigma_4$ &quaternary alphabet, i.e. $q=4, \Sigma_4=\{0,1,2,3\}$\\
 $\D$& DNA alphabet, $\D=\{\tA, \tT, \tC, \tG\}$\\
 $\bx \by$& the concatenation of two sequences \\
 $\bx || \by$ & the interleaved sequence \\
 $\bsg, \bU_\bsg, \bL_\bsg$ & a DNA sequence $\bsg$, the upper sequence of $\bsg$, and the lower sequence of $\bsg$ \\
  $\Psi$  & the one-to-one map that converts a DNA sequence to a binary sequence \\
  ${\rm Syn}(\bx)$& the syndrome of a sequence $\bx$\\
  ${\rm indel} $ & single insertion or single deletion \\
  ${\rm edit} $ & single insertion, or single deletion, or single substitution\\
  $\Bindel(\bx)$   & the set of words that can be obtained from $\bx$ via at most a single indel \\
  $\Bedit(\bx)$    &  the set of words that can be obtained from $\bx$ via at most a single edit\\
\hline
\end{tabular}
\vspace{5mm}

\begin{tabular}{ p{3cm} p{7cm} p{5.5cm} p{1.5cm}}
\hline
  Encoder / Decoder   & Description & Redundancy & Remark \\
 \hline
  
  
  
  
  
  
  $\enc_{\rm RLL}^{A}, \dec_{\rm RLL}^{A}$& encoder and decoder for $\ell$-runlength limited codes using enumeration technique & ${\bf r}_A=n-\floor{\log_4|\C(n,\ell,q)|}$ (symbols)  &Section III-A\\
  $\enc_{\rm RLL}^{B}, \dec_{\rm RLL}^{B}$& encoder and decoder for $\ell$-runlength limited codes using sequence replacement technique & 
         ${\bf r}_B=1 (symbol)$ if $n\le (q-1)q^{\ell-1}+\ell-1$ or 
         $\ceil{n/((q-2)q^{\ell-1}+\ell)}$ symbols, otherwise 
                 &Section III-B\\
 \hline                 
  $\enc_{\epsilon\tG\tC}^{C}, \dec_{\epsilon\tG\tC}^{C}$& encoder and decoder for $\epsilon$-balanced quaternary codes using binary template & ${\bf r}_C=\ceil{\log_2 \left( \floor{1/2\epsilon}+1 \right)}$ (bits)  &Section IV-C\\
  $\enc_{\epsilon\tG\tC}^{D}, \dec_{\epsilon\tG\tC}^{D}$& encoder and decoder for $\epsilon$-balanced quaternary codes using Knuth's technique & ${\bf r}_D=2\ceil{\log_4 \left( \floor{1/2\epsilon}+1 \right)}$ (symbols)  &Section IV-D\\
  \hline
$\enc_{(\epsilon,\ell)}, \dec_{(\epsilon,\ell)}$& constrained encoder/decoder for $\epsilon$-balanced and $\ell$-runlength limited codes  & ${\bf r}_A+{\bf r}_D+4$ (symbols) or ${\bf r}_B+{\bf r}_D+4$ (symbols) &Section V\\
\hline
$\enc_{(\epsilon,\ell;\Bindel)}$, $\dec_{(\epsilon,\ell;\Bindel)}$& error-control encoder/decoder for $\epsilon$-balanced and $\ell$-runlength limited codes that can correct an indel  & ${\bf r}_A+{\bf r}_D+\log_2 n+\Theta(1)$ (symbols) or ${\bf r}_B+{\bf r}_D+\log_2 n+\Theta(1)$ (symbols) &Section VI-B\\

$\enc_{(\epsilon,\ell;\Bedit)}$, $\dec_{(\epsilon,\ell;\Bedit)}$& error-control encoder/decoder for $\epsilon$-balanced and $\ell$-runlength limited codes that can correct an edit  & ${\bf r}_A+{\bf r}_D+2\log_2 n+\Theta(1)$ (symbols) or ${\bf r}_B+{\bf r}_D+2\log_2 n+\Theta(1)$ (symbols) &Section VI-C\\
  \hline
 \end{tabular}
\caption{Notation and Results Summary. The redundancy is computed for DNA codewords of length $n$, given $\ell, \epsilon>0$. }\label{tab.notation}
\end{table}


\section{Notation}\label{sec:prelim}
Let $\Sigma_q=\{0,1,2,\ldots,q-1\}$ denote an {\em alphabet} of size $q\geq2$. Particularly, when $q=4$, we use the following relation $\Phi$ between the decimal alphabet $\Sigma_4=\{0,1,2,3\}$ and the nucleotides $\D=\{{\tt A},{\tt T},{\tt C},{\tt G}\}$, $\Phi: 0 \to {\tt A}, 1 \to {\tt T}, 2 \to {\tt C}$, and $3 \to {\tt G}$.


Given two sequences $\bx$ and $\by$, we let $\bx\by$ denote the {\em concatenation} of the two sequences.
In the special case where $\bx,\by \in \Sigma_q^n$, we use $\bx || \by$ to denote their {\em interleaved sequence} $x_1y_1x_2y_2\ldots x_ny_n$.

Let $\bsg=\sigma_1\sigma_2\ldots\sigma_n \in\Sigma_4^n$, denote a 4-ary strand of $n$ nucleotides. The ${\tt G}{\tt C}$-content or weight of strand $\bsg$, denoted by $\omega(\bsg)$, is defined by $\omega(\bsg)=(1/n).\sum_{i=1}^n \varphi(\sigma_i)$ where $\varphi(\sigma_i)=0$ if $\sigma_i\in\{0,1\}$ and $\varphi(\sigma_i)=1$ if $\sigma_i\in\{2,3\}$. Given $\epsilon>0$, we say that $\bsg$ is \emph{$\epsilon$-balanced} if $|\omega(\bsg)-0.5|\leq \epsilon$, in other words, $\omega(\bsg) \in (0.5-\epsilon,0.5+\epsilon)$. In particular, when $n$ is even and $\epsilon=0$, we say $\bsg$ is {\em ${\tt G}{\tt C}$-balance}. Over binary alphabet, a vector $\bx\in \{0,1\}^n$ is called {\em balanced} if the number of ones in $\bx$, or the weight ${\rm wt}(\bx)$, is $n/2$.

On the other hand, given $\ell>0$, we say that $\bsg$ is \emph{$\ell$-runlength limited} if any run of the same nucleotide is at most $\ell$. For DNA-based storage, we are interested in codewords that are $\epsilon$-balanced and $\ell$-runlength limited for sufficient small $\epsilon=o(1)$, $\ell=o(n)$.

\begin{definition}
A nucleotide encoder $\enc:\{0,1\}^m \to\Sigma_4^n$ is a {\em $(\epsilon,\ell)$-constrained encoder}
if $\enc(\bx)$ is $\epsilon$-balanced and $\ell$-runlength limited for all $\bx\in\{0,1\}^m$.
\end{definition}

Motivated by the error behavior in DNA storage, we investigate constrained codes that also have error-correction capability. Such codes are referred as {\em error-control-codes}. We use $\B$ to denote the {\em error ball} function. For a sequence $\bx\in \Sigma_4^n$, let $\B^{\rm D}(\bx)$, $\B^{\rm I}(\bx)$, and $\B^{\rm S}(\bx)$ denote the set of all words obtained from $\bx$ via a single deletion, single insertion, or at most one substitution, respectively, and set

\begin{equation*}
\B^{\rm indel}(\bx)\triangleq\B^{\rm I}(\bx)\cup \B^{\rm D}(\bx), \B^{\rm edit}(\bx) \triangleq \B^{\rm S}(\bx)\cup \B^{\rm I}(\bx)\cup \B^{\rm D}(\bx).
\end{equation*}

Observe that when $\bsg\in\Sigma_4^n$, both $\Bindel(\bsg)$ and $\Bedit(\bsg)$ are subsets of $\Sigma_4^{n-1}\cup \Sigma_4^n \cup \Sigma_4^{n+1}$.
Hence, for convenience, we use $\Sigma_4^{n*}$ to denote the set $\Sigma_4^{n-1}\cup \Sigma_4^n \cup \Sigma_4^{n+1}$.

\begin{definition} Let $\C\subseteq \Sigma_4^n$. Given $\epsilon,\ell>0$ and the error ball function $\B$, we say that $\C$ is an $(\epsilon,\ell; \B)$-error control codes if
\begin{enumerate}[(i)]
\item For all $\bc \in \C$, $\bc$ is $\epsilon$-balanced,
\item For all $\bc \in \C$, $\bc$ is $\ell$-runlength limited, and
\item $\B(\bc)\cap \B(\bc')=\varnothing$ for all distinct $\bc,\bc' \in \C$.
\end{enumerate}
\end{definition}

For a code $\C \subseteq \Sigma_q^n$, the rate of $\C$, denoted by ${\rm rate}_{\C}$, is defined by ${\rm rate}_{\C} \triangleq (1/n) \log_q |\C|$. The asymptotic rate of the family of codes $\{\C(n,N;q)\}_{n=1}^{\infty}$ is defined by $\lim_{n\to\infty} (1/n) \log_q |\C|$, if the limit exists.

\begin{definition}
A nucleotide encoder $\enc: \{0,1\}^m\to \Sigma_4^n$
is an $(\epsilon,\ell; \B)$-{\em error-control-encoder} if $\enc(\bx)$ is $\epsilon$-balanced and $\ell$-runlength limited for all $\bx\in\{0,1\}^m$, furthermore there exists a {\em decoder} map $\dec:\Sigma_4^{n*} \to \{0,1\}^m$ such that the following hold.
\begin{enumerate}[(i)]
\item For all $\bx\in\{0,1\}^n$, we have $\dec\circ\enc(\bx)=\bx$.
\item If $\bc=\enc(\bx)$ and $\bc'\in \B(\bc)$, then $\dec(\bc')=\bx$.
\end{enumerate}
\end{definition}

Hence, we have that the code $\C=\{\bc : \bc=\enc(\bx),\, \bx\in\{0,1\}^m\}$ and hence, $|\C|=2^m$.
The {\em redundancy of the encoder} is measured by the value $2n -  m$ (in bits) or $n-m/2$ (nucleotide symbols).

\section{Efficient Homopolymer Runlength Limited Codes}
\label{sec:RLL}

We present two methods of constructing maximum runlength limited $q$-ary constrained codes. Method A uses enumerative coding technique to rank/unrank all codewords. While the technique is standard in constrained coding and combinatorics literature, our contribution is a detailed analysis of the space and time complexities of the respective algorithm. The encoder achieves maximum code rate, for example, when $\ell=3, n=200, q=4$, the rate of the encoder is 1.98 bits/nt. However, the time and space complexity is $O(n^2)$, which makes it less attractive than the sequence replacement technique in Method B.

\subsection{Method A Based on Enumeration Coding}
Let $\C(n,\ell,q)$ denote the set of all $q$-ary $\ell$-runlength limited sequences of length $n$. We first obtain a recursive formula for the size of $\C(n,\ell,q)$. This recursive formula is useful in the development of the ranking/unranking methods. To this end, we partition $\C(n,\ell,q)$ into $\ell$ classes and provide bijections from $q$-ary $\ell$-runlength limited sequences of
shorter lengths into them. For $1\leq i\leq \ell$, let $\C_i(n,\ell,q)$ denote the set of all $q$-ary $\ell$-runlength limited sequences of length $n$ whose suffix is the repetition of a symbol in $\Sigma_q$ for exactly $i$ times. Clearly, we have $\C_i(n,\ell,q) \cap \C_j(n,\ell,q)=\varnothing$ for $i\neq j$ and
\begin{equation*}
\C(n,\ell,q) = \bigcup_{i=1}^{\ell} \C_i(n,\ell,q)
\end{equation*}

Let $[n]$ denote the set $\{1,2,\ldots,n\}$.
 Consider $\ell$ maps $\phi_1,\phi_2,\ldots,\phi_{\ell}$ where
\begin{equation*}
\phi_i: \C(n-i,\ell,q) \times [q-1] \to \C_i(n,\ell,q), \mbox{ for $1\leq i\leq \ell$}.
\end{equation*}
If $\bx=x_1x_2\ldots x_{n-i}\in \C(n-i,\ell,q)$ and $j\in [q-1]$,
set $a$ to be the $j$th element in $\Sigma_q\setminus \{x_{n-i}\}$.
Then set $\phi_i(\bx,j)=x_1x_2\ldots x_{n-i} a^i$. Here, $a^i$ denotes the repetition of symbol $a$ for $i$ times.

\begin{theorem}\label{recurrence}
For $1\leq i\leq \ell$, the map $\phi_i$ is a bijection. We then have the following recursion. For $1\leq n\leq \ell$, $|\C(n,\ell,q)|=q^{\ell}$, and for $n > \ell$
\begin{equation*}
|\C(n,\ell,q)| = \sum_{i=1}^{\ell} (q-1) |\C(n-i,\ell,q)|.
\end{equation*}
Therefore, ${\rm rate}_{\C(n,\ell,q)}=\log_q \lambda$,
where $\lambda$ is the largest real root of equation $x^{\ell}-\sum_{i=0}^{\ell-1} (q-1)x^i=0$.
\end{theorem}

\begin{proof}
We can prove that $\phi_i$ is bijection for $1\le i\le \ell$ by constructing the inverse map $\phi^{-1}_i$. Specifically, we set $\phi_i^{-1}: \C_i(n,\ell,q) \to \C(n-i,\ell,q) \times [q-1]$ such that for $\bx=x_1x_2\ldots x_n \in \C_i(n,\ell,q), \phi_i^{-1}(\bx)=(x_1\ldots x_{n-i},j)$ where $j$ is the index of $x_n$ in $\Sigma_q\setminus \{x_{n-i}\}$. It can be verified that $\phi_i\circ \phi_i^{-1}$ and $\phi_i^{-1}\circ \phi_i$ are identity maps on their respective domains. Since $\C(n,\ell,q) = \bigcup_{i=1}^{\ell} \C_i(n,\ell,q)$, we then have for $n > \ell$
\begin{equation*}
|\C(n,\ell,q)| = \sum_{i=1}^{\ell} (q-1) |\C(n-i,\ell,q)|.  \qedhere
\end{equation*}
\end{proof}

We now construct the RLL-Encoder A by providing a method of ranking/unranking all codewords in $\C(n,\ell,q)$. A {\em ranking function} for a finite set $S$ of cardinality $N$ is a bijection
${\rm rank}:S\rightarrow [N]$. Associated with the function {\rm rank} is a unique {\em unranking function}
${\rm unrank}:[N]\rightarrow S$, such that ${\rm rank}(s)=j$ if and only if ${\rm unrank}(j)=s$ for all $s\in S$ and
$j\in[N]$.

The basis of our ranking and unranking algorithms is the bijections $\{\phi_i\}_{i=1}^{\ell}$ defined earlier. As implied by the codomains of these maps, for $n> \ell$, we order the words in $\C(n,\ell,q)$ such that words in $\C_i(n,\ell,q)$ are ordered before words in $\C_j(n,\ell,q)$ for $i<j$.
For words in $\C(n,\ell,q)$ where $n\leq\ell$, we simply order them lexicographically.
We illustrate the idea behind the unranking algorithm through an example.

\begin{example}
Let $n=5, q=4, \ell=3$. We then have $|\C(n,3,4)|=3|\C(n-1,3,4)|+3|\C(n-2,3,4)|+3|\C(n-3,3,4)|$ and the values of $\C(m,\ell,q)$ are as follow.
\begin{center}
\begin{tabular}{|c||c|c|c|c|c|}
\hline
$m$ & 1 & 2 & 3 & 4 & 5 \\ \hline
$I_{{\le}2}(m,q)$ & 4 & 16 & 64 & 252 & 996\\ \hline
\end{tabular}
\end{center}
Suppose we want to compute the 900th codeword $\bc \in \C(5,3,4)$, in other words, ${\rm unrank}(900)$. We have
\begin{align*}
&\C(5,3,4) = {\color{blue}{\C_1(5,3,4)}} \cup {\color{red}{\C_2(5,3,4)}} \cup {\color{green}{\C_3(5,3,4)}} = \\
& {\color{blue}{\phi_1(\C(4,3,4)\times [3])}} \cup {\color{red}{\phi_2(\C(3,3,4)\times[3])}} \cup {\color{green}{\phi_3(\C(2,3,4)\times[3])}},
\end{align*}
Since $900>3 |\C(4,3,4)|=756$ and $900< 3 |\C(4,3,4)|+ 3|\C(3,3,4)|=948$, the 900th codeword of $\C(5,3,4)$, which is the $900-756=144$th codeword in {\color{red}{$\C_2(5,3,4)$}}, is the image of map ${\color{red}{\phi_2}}$. Since $144= 3 \times 48 + 0$, the construction of $\phi_2$ tells us that the 144th codeword in {\color{red}{$\C_2(5,3,4)$}} is the image of the $48$th codeword,  {\color{red}{$\bx \in \C(3,3,4)$}} under ${\color{red}{\phi_2}}$.
The $48$th word of {\color{red}{$\C(3,3,4)$}} is $344$. Hence, $\bc=\phi_2(\bx,3)$ This gives
\begin{align*}
{\rm unrank}(900) &= \phi_2(344,3) \\
&= 34433
\end{align*}
\end{example}
The formal unranking/ranking algorithms are described in Algorithm~\ref{alg1} and Algorithm~\ref{alg2}.

\begin{algorithm}
\small
\caption{${\tt unrank}(n,\ell,q,M)$}\label{alg1}
\begin{algorithmic}
\REQUIRE Integers $n \geq 1$, $\ell \geq 1$, $q\ge 2$, $1\leq M\leq |\C(n,\ell,q)|$
\ENSURE $\bc$, where $\bc$ is the codeword of rank $M$ in $\C(n,\ell,q)$
\vspace{0.05in}
\IF{$n\leq \ell$} \RETURN{$M$th codeword in $\C(n,\ell,q)$} \ENDIF

{\bf Search the first index} $1\leq j\leq\ell$ such that
\begin{equation*}
M \leq \sum_{i=1}^{j} (q-1)|\C(n-i,\ell,q)|
\end{equation*}

\STATE{$M'\gets \sum_{i=1}^{j} (q-1)|\C(n-i,\ell,q)|-M$}
\vspace{0.05in}

\STATE{$M''\gets \ceil*{M'/(q-1)}$}
\vspace{0.05in}

\STATE{$k\gets M'\pmod{q-1}$}

\RETURN{$\phi_j({\tt unrank}(n-j,\ell,q,M''),k)$}

\end{algorithmic}
\end{algorithm}


\begin{algorithm}
\small
\caption{${\tt rank}(n,\ell,q,\bc)$}\label{alg2}
\begin{algorithmic}
\REQUIRE $n \geq 1, \ell\geq 1$, $q\ge 2$ and codeword $\bc=c_1c_2\ldots c_n$
\ENSURE $M$, where $1 \leq M \leq |\C(n,\ell,q)|$, the rank of $\bc$ in $\C(n,\ell,q)$
\vspace{0.05in}
\IF{$n\leq \ell$} \RETURN{${\tt rank}(\bc)$ in $\C(n,\ell,q)$} \ENDIF

\IF{the suffix of $\bc$ is the repetition of symbol $a$ for $i$ times}
\STATE{$\bc'\gets c_1c_2\ldots c_{n-i}$}
\STATE{$i\gets$ the index of $a$ in $\Sigma_q\setminus\{c_{n-i}\}$}

\RETURN{$({\tt rank}(n-i,\ell,q,\bc')-1)(q-1)+i+\sum_{j=1}^{i-1} (q-1)|\C(n-j, \ell, q)|$}
\ENDIF

\end{algorithmic}
\end{algorithm}

\begin{example}
Let $n=5, \ell=3$ and $q=4$ as before.
Suppose we want to compute ${\tt rank}(34433)$.
Since $34433\in{\color{red}{\C_2(5,3,4)}}$,
we have that $34433$ is obtained from applying $\phi_2$ to $344\in {\color{red}{\C(3,3,4)}}$.
The adding symbol is 3, which is the third element in $\Sigma_4\setminus\{4\}$. Therefore,
\begin{align*}
{\rm rank}(34411)&= 3 {\color{blue}{|\C(4,3,4)|}}+3 ({\rm rank}(344)-1)+3\\
&= 3 \times 252+3 \times 47 + 3 \\
&=900.
\end{align*}
\end{example}

\noindent The set of values of $\{|\C(m,\ell,q)|: m\le n\}$ required in Algorithms~\ref{alg1} and~\ref{alg2}
can be precomputed based on the recurrence in Theorem~\ref{recurrence}.
Since the size of $\C(n,\ell,q)$ grow exponentially, these $n$ stored values require $O(n^2)$ space.

Next, Algorithms~\ref{alg1} and~\ref{alg2} involve $O(n)$ iterations and
each iteration involves a constant number of arithmetic operations.
Therefore, Algorithms~\ref{alg1} and~\ref{alg2} involve $O(n)$ arithmetics operations and
have time complexity $O(n^2)$. For completeness, we summarize the RLL-Encoder A and RLL-Decoder A as follows.

\vspace{0.05in}

\noindent{\bf RLL-Encoder A}. Set $m=\floor{\log_2|\C(n,\ell,q)|}$.

{\sc Input}: $\bx \in \{0,1\}^m$\\
{\sc Output}: $\bc \triangleq \enc_{\rm RLL}^{A}(\bx)\in\C(n,\ell,q)$ \\[-2mm]

\begin{enumerate}[(I)]
\item Let $M$ be the positive integer whose binary representation of length $m$ is $\bx$.
\item Use Algorithm~\ref{alg1}, set $\bc={\tt unrank}(n,\ell,q,M)$.
\item Output $\bc$.
\end{enumerate}

\vspace{0.05in}

\noindent{\bf RLL-Decoder A}. Set $m=\floor{\log_2|\C(n,\ell,q)|}$.

{\sc Input}: $\bc \in \C(n,\ell,q)$\\
{\sc Output}: $\bx \triangleq \dec_{\rm RLL}^{A}(\bc)\in\{0,1\}^m$ \\[-2mm]

\begin{enumerate}[(I)]
\item Use Algorithm~\ref{alg2}, set $M={\tt rank}(n,\ell,q,\bc)$.
\item Let $\bx$ be the binary representation of length $m$ of $M$.
\item Output $\bx$.
\end{enumerate}

\subsection{Method B Based on Sequence Replacement Technique}
The sequence replacement technique has been widely used in the literature \cite{immink2018, W2010, schoeny2017, O2019}. This is an efficient method for removing forbidden substrings from a source word. In general, the encoder removes the forbidden strings and subsequently inserts its representation (which also includes the position of the substring) at predefined positions in the sequence. For example, Schoeny \et{}\cite{schoeny2017} used only one redundant bit to encode RLL binary sequences with $\ell\ge\ceil{\log n}+3$. However, for DNA data storage, with $n\in[100,200]$, it is normally required that $\ell\le6$. Recently, Immink \et{}\cite{immink2018} described a simple method for constructing $\ell$-runlength limited $q$-ary codes. However, the required codeword length $n$ is bounded by a function of $\ell$ and $q$. For example, when $\ell=3$, the method is only applicable for $n\le39$ (refer to \cite[Table II]{immink2018}). In this work, we show that such bound can be improved, and hence, the redundancy can be further reduced. For DNA storage channel, when $n\le200$, $\ell\in\{5,6\}$, our encoder incurs only one redundant symbol.

\begin{definition} For a sequence $\bx=x_1 x_2 \ldots x_n \in \Sigma_q^{n}$, the {\em differential of $\bx$}, denoted by ${\rm Diff(\bx)}$, is a sequence $\by=y_1 y_2 \ldots y_n \in \Sigma_q^{n}$, where $y_1=x_1$ and $y_i = x_i-x_{i-1} \ppmod{q}$ for $2\leq i\leq n$.
\end{definition}

It is easy to see that from $\by=y_1y_2\ldots y_n= {\rm Diff(\bx)}$, we can determine $\bx$ uniquely as $x_i=\sum_{j=1}^{i} y_j \ppmod{q} $ for $1\leq i\leq n$.  For convenience, we write $\bx={\rm Diff^{-1}(\by)} $.


\begin{lemma}\label{relation}
Let $\bx \in \Sigma_q^{n}$. If the longest run of zero in ${\rm Diff(\bx)}$ is at most $\ell-1$ then $\bx$ is $\ell$-runlength limited.
\end{lemma}

We now present an efficient encoder for $\ell$-runlength limited $q$-ary codes, and refer this as RLL Encoder B or $\enc_{\rm RLL}^B$. For a source data $\bx \in \Sigma_q^{N-1}$, we encode $\by=\enc(\bx) \in \Sigma_q^{N}$ such that $\by$ contains no $0^{\ell}$ as a substring, and then output $\bc={\rm Diff}^{-1}(\by)$. 
\vspace{0.01in}

\noindent{\em Initial Step}. The encoder simply appends a `0' to the end of $\bx$, yielding the $N$-symbols word, $\bx0$. The encoder then checks the word $\bx0$, and if there is no substring $0^{\ell}$, the output is simply $\bc=\bx0$. Otherwise, it proceeds to the replacement step.
\vspace{0.05in}

\noindent{\em Replacement Procedure}. Let the current word $\bc=\by 0^{\ell} \bz$, where, by assumption, the prefix $\by$ has no forbidden $0^{\ell}$ and the run $0^{\ell}$ starts at position $p$, where $1\leq p\leq N-\ell$. The encoder removes $0^{\ell}$ and updates the current word to be $\bc=\by\bz {\bf R} e$, where the {\em pointer} ${\bf R} e$ is used to represent the position $p$, and
\begin{enumerate}[(i)]
\item ${\bf R} \in \Sigma_q^{\ell-1}$,
\item $e \in \Sigma_q \setminus \{0\}$,
\end{enumerate}
Note that the number of unique combinations of the pointer ${\bf R} e$ equals $(q-1)q^{\ell-1}$. Note that the current word $\bc=\by\bz {\bf R} e$ is of length $N$. If, after the replacement, $\bc$ contains no substring $0^{\ell}$ then the encoder returns $\bc$ as the codeword. Otherwise, the encoder repeats the replacement procedure for the current word $\bc$ until all substrings $0^{\ell}$ have been removed. Noted that during every step, the length of the codeword is preserved. Since the last symbol in any additional pointer is nonzero, the concatenation between any two consecutive pointers ${\bf R_1} e_1 {\bf R_2} e_2$ does not produce any substring $0^{\ell}$, this procedure is guarantee to terminate. As the position $p$ is in the range $1\leq p\leq N-\ell+1$, and the number of combinations of ${\bf R} e$ equals $(q-1)q^{\ell-1}$, we conclude that $N$ is upper bounded by
\begin{equation}\label{UB1}
N \leq (q-1)q^{\ell-1}+\ell-1, \mbox{ for } \ell \geq 2.
\end{equation}

\noindent{\em Decoding Procedure}. The decoder checks from the right to the left. If the last symbol is `0', the decoder simply removes the symbol `0' and identifies the first $N-1$ symbols are source data. On the other hand, if the last symbol is not `0', the decoder takes the suffix of length $\ell$, identifies it is the pointer, and then adds back the substring $0^{\ell}$ accordingly. It terminates when the first symbol `0' is found.
\begin{remark}
The bound in \eqref{UB1} implies that for $q=4,\ell\in\{4,5,6\}$,  our encoder uses only one redundant symbols for all $n\le196$. Table~\ref{compare} shows the improvement with respect to the result provided in \cite{immink2018}. In addition, this algorithm can be easily extended for the case of arbitrary length $n\gg N$. The main idea is that we divides the source data into subwords of length $N-1$, encodes separately each subword and concatenate them. The representation pointer needs to be modified so that the concatenation between any two encoded subwords does not contain a substring $0^{\ell}$. To do so, we simply append '1' to the end of the source data instead, and require the pointers of the form ${\bf R}e$ where ${\bf R}\in \Sigma_q^{\ell-1}$ and $e\notin\{0,1\}$. The replacement procedure and decoding procedure can be proceeded similarly. 
\end{remark}


\begin{table}[h!]
\centering 
 \begin{tabular}{|c|| c| c|}
 \hline
 $\ell~\backslash~n_{\rm max}$ & Bound in \eqref{UB1} &  Previous work \cite{immink2018} \\[1ex]
 \hline
 2  &   13 &11\\
 3  &   50 &39\\
 4  &   195 & 148\\
 5  &    772 & 581\\
\hline
\end{tabular}
\caption{Maximum length $n$ that an encoder can achieve the rate $(n-1)/n$ for $\ell$-runlength limited quaternary codes.}
\label{compare}
\end{table}

\section{Efficient $\tG\tC$-Content Constrained Codes}
\label{sec:GC}
 In this section, we propose linear-time encoders/decoders that translate binary input data to DNA strands whose $\tG\tC$-content is within $[0.5 - \epsilon, 0.5 + \epsilon]$ for arbitrary $\epsilon > 0$, with fixed number of redundant bits. This method yields a significant improvement in coding redundancy with respect to the prior works. We first review the Knuth's balancing technique.
 \subsection{Knuth's Balancing Technique}
Knuth's balancing technique is a linear-time algorithm that maps a binary message $\bx$ to a balanced word $\by$ of the same length by flipping the first $t$ bits of $\bx$ \cite{knuth}. The crucial observation demonstrated by Knuth is that such an index $t$ always exists and $t$ is commonly referred to as the {\em balancing index}. To represent the balancing index, Knuth appends $\by$ with a short balanced suffix of length $\log n$ and so, a lookup table of size $\log n$ is required.

Several works in the literature used this technique to encode DNA strands whose $\tG\tC$-content is exactly balanced (for example, \cite{dube2019, tt2019}), and the coding redundancy is approximately $\log n$. We generalize this technique for binary codes first.

\subsection{Generalization of Knuth's Balancing Technique}

\begin{definition}
Let $n$ be even. For arbitrary $\epsilon>0$, a binary word $\bx\in\{0,1\}^n$ is {\em $\epsilon$-balanced} if the weight of $\bx$, ${\rm wt}(\bx)$, satisfies
\begin{equation*}
\left| \frac{{\rm wt}(\bx)}{n}-0.5 \right| \leq \epsilon.
\end{equation*}
In other words, we have $ 0.5n-\epsilon n \leq{\rm wt}(\bx) \leq 0.5n+\epsilon n$.
\end{definition}

\begin{definition}
Let $n$ be even. For arbitrary $\epsilon>0$, the index $t$, where $1\leq t\leq n$, is called the {\em $\epsilon$-balanced index} of $\bx\in\{0,1\}^n$ if the word $\by$ obtained by flipping the first $t$ bits in $\bx$ is {\em $\epsilon$-balanced}.
\end{definition}

We now show that  such an index $t$ always exists and there is an efficient method to find $t$. For $n$ even, let the {\em $\epsilon$-balanced set} ${\rm S}_{\epsilon,n} \subset \{0,1,2,\ldots,n\}$ be the set of the following indices.
\begin{equation}\label{set}
{\rm S}_{\epsilon,n}= \{0,n\} \cup \{2\floor{\epsilon n}, 4\floor{\epsilon n}, 6\floor{\epsilon n}, \ldots \}.
\end{equation}

The size of ${\rm S}_{\epsilon,n}$ is at most $\floor{1/2\epsilon}+1$.

\begin{theorem}\label{epsilon-balanced}
Let $n$ be even, $\epsilon>0$. For arbitrary binary sequence $\bx \in \{0,1\}^{n}$, there exists an index $t$ in the set ${\rm S}_{\epsilon,n}$, such that $t$ is the $\epsilon$-balanced index of $\bx$.
\end{theorem}

\begin{proof}
In the trivial case, when $\bx$ is $\epsilon$-balanced, the index $t=0$, which is in the set ${\rm S}_{\epsilon,n}$. Assume that $\bx$ is not $\epsilon$-balanced, and without loss of generality, assume that $ {\rm wt}(\bx) < 0.5n-\epsilon n.$ Let ${\rm Flip}_k(\bx)$ be the word obtained by flipping the first $k$ bits in $\bx$. Since ${\rm wt}(\bx) < 0.5n-\epsilon n$, we have ${\rm wt}({\rm Flip}_n(\bx)) > 0.5n+\epsilon n$. Now consider the list of indices that we try to obtain an $\epsilon$-balanced word, $t_1=2\floor{\epsilon n}, t_2=4\floor{\epsilon n}$, and so on. Since ${\rm Flip}_{t_i}(\bx)$ and  ${\rm Flip}_{t_{i+1}}(\bx)$ differ at at most $2\epsilon n$ positions, and ${\rm wt}(\bx) < 0.5n-\epsilon n$, ${\rm wt}({\rm Flip}_n(\bx)) > 0.5n+\epsilon n$, there must be an index $t$ such that $ 0.5n-\epsilon n \leq{\rm wt}({\rm Flip}_t(\bx)) \leq 0.5n+\epsilon n$.
\end{proof}

We provide two methods to construct $\tG\tC$-Content constrained codes. The first method uses $\epsilon$-balanced binary codes as a template to construct $\epsilon$-balanced quaternary codes with at most $\log \left( \floor{1/2\epsilon}+1 \right)$ bits of redundancy. On the other hand, the second method proceeds directly over quaternary alphabet and appends a short balanced suffix to the end of each codeword to indicate the $\epsilon$-balanced index.


\subsection{Binary Construction of $\tG\tC$-Content Constrained Codes}

When $q=4$, we consider the following one-to-one correspondence between quaternary alphabet and two-bit sequences:
\[ 0 \leftrightarrow 00,\quad 1 \leftrightarrow 01,\quad 2 \leftrightarrow 10,\quad 3 \leftrightarrow 11.\]
Therefore, given a DNA sequence $\bsg$ of length $n$, we have a corresponding binary sequence $\bx\in\{0,1\}^{2n}$
and we write $\bx=\Psi(\bsg)$ or $\bsg=\Psi^{-1}(\bx)$. Given $\bsg \in \Sigma_4^n$, let $\bx=\Psi(\bsg)\in \{0,1\}^{2n}$
and we set $\bU_\bsg=x_1x_3\cdots x_{2n-1}$ and $\bL_\bsg=x_2x_4\cdots x_{2n}$.
In other words, $\bsg=\Psi^{-1}(\bU_\bsg || \bL_\bsg)$.
We refer to $\bU_{\sigma}$ and $\bL_\bsg$ as the {\em upper sequence} and {\em lower sequence} of $\bsg$, respectively. The following result is immediate.

\begin{lemma}
Let $\bsg \in \Sigma_4^n$. We have $\bsg$ is $\epsilon$-balanced if and only if $\bU_{\sigma}$ is $\epsilon$-balanced.
\end{lemma}

\noindent{\bf $\epsilon$\GC-Encoder C}. Given $n,\epsilon>0$, set $k=\ceil{\log \left( \floor{1/2\epsilon}+1 \right)}$ and $m=2n-k$. Set ${\rm S}_{\epsilon,n}$ be the set of indices as constructed in \eqref{set} and we construct a one-to-one correspondence between the indices in ${\rm S}_{\epsilon,n}$ and $k$ bits sequences.

{\sc Input}: $\bx\in \{0,1\}^n$, $\by\in \{0,1\}^{n-k}$ and so, $\bx\by\in\{0,1\}^m$\\
{\sc Output}: $\bsg = \enc_{\epsilon\tG\tC}^{C}(\bx\by)$\\[-3mm]

\begin{enumerate}[(I)]
\item Search for the first $t$ in ${\rm S}_{\epsilon,n}$, such that ${\rm Flip}_t(\bx)$ is $\epsilon$-balanced.
\item Set $\bx'={\rm Flip}_t(\bx)$.
\item Let $\bz$ be the $k$ bits sequence representing index $t$.
\item Set $\by'=\by \bz$ of length $n$
\item Finally, we set $\bsg\triangleq \Psi^{-1}(\bx'||\by')$.
\end{enumerate}

\begin{example}
Let $n=10, \epsilon=0.1, k=\ceil{\log \left( \floor{1/2\epsilon}+1 \right)}=3$, i.e. we want the $\tG\tC$-content of each codeword is within $[0.4,0.6]$. The set ${\rm S}_{\epsilon,n}=\{0,2,4,6,8,10\}$ is of size six. We construct the one-to-one correspondence between the indices and $3$ bits sequences: $ 0 \to 000, 2 \to 001, 4 \to 010, 6 \to 100, 8 \to 011$ and $10 \to 111$. Suppose the input sequence is $\bc=0^{17}$, i.e $\bx=0^{10}$ and $\by=0^{7}$. We find the index $t=4$. Follow the encoder, we get $\bx'=1111000000$ and $\by'=0000000{\color{blue}{010}}$. We then obtain $\bsg=\Psi^{-1}(\bx'||\by')=2 2 2 2 0 0 0 0 1 0$.
\end{example}
\vspace{0.05in}

\noindent{\bf $\epsilon$\GC-Decoder C}. Given $n,\epsilon>0$, set $k=\ceil{\log \left( \floor{1/2\epsilon}+1 \right)}$ and $m=2n-k$.

{\sc Input}: $\bsg\in \Sigma_4^n$, $\bsg$ is $\epsilon$-balanced\\
{\sc Output}: $\bx\by\in\{0,1\}^m$\\[-3mm]

\begin{enumerate}[(I)]
\item Set $\bx'=\bU_{\sigma} \in \{0,1\}^n$ and $\by'=\bL_\bsg \in \{0,1\}^n$.
\item Set $\bz$ be the suffix of length $k$ in $\by'$ and let $t$ be the index in ${\rm S}_{\epsilon,n}$ corresponding to $\bz$.
\item Set $\bx={\rm Flip}_t(\bx')$.
\item Set $\by=\by'$ removes $\bz$
\item Finally, we output $\bx \by$.
\end{enumerate}
\vspace{0.05in}

\begin{remark}
For constant $\epsilon>0$, the complexity of an $\epsilon$\GC-Encoder C is linear and the redundancy is constant. For example,
when $n=200, \epsilon=0.1$, i.e. the $\tG\tC$-content is within $[0.4,0.6]$, the set ${\rm S}_{\epsilon,n}=\{0,40,80,120, 160, 200\}$ is of size six. The $\epsilon$\GC-Encoder C uses only $\ceil{\log 6}=3$ bits of redundancy to indicate the $\epsilon$-balanced index in the lower sequence and the rate of the encoder is $1.985$ bits/nt. Similarly, when $\epsilon=0.05$,  i.e. the $\tG\tC$-content is within $[0.45,0.55]$, the $\epsilon$\GC-Encoder C uses only $\ceil{\log 11}=4$ bits of redundancy and the rate is $1.98$ bis/nt.
\end{remark}


\subsection{Knuth-like Construction of $\tG\tC$-Content Constrained Codes}
Consider the quaternary alphabet $\Sigma_4=\{0,1,2,3\}$. To apply Knuth's method, we define the flipping rule $f : \Sigma_4 \to  \Sigma_4$, where
$f(0) = 2,f(2) = 0,f(1) = 3$ and $f(3) = 1$. For a sequence $\bsg \in \Sigma_4^n$ and index $i$ with $0\leq i\leq n$,
$f_i(\bsg)$ denotes the sequence obtained by flipping the first $i$ symbols of $\bsg$ under $f$.

\begin{definition}
Let $n$ be even. For arbitrary $\epsilon>0$, the index $t$, where $1\leq t\leq n$, is called the {\em $\epsilon$-balanced index} of $\bsg\in\Sigma_4^n$ if the sequence $\bsg'=f_t(\bsg)$ is {\em $\epsilon$-balanced}.
\end{definition}

\begin{example}
Consider $n=10, \epsilon=0.1$. Let $\bsg=0000000000$.
Observe that $f_4(\bsg)={\tt {\color{red}2222}000000}$, $f_5(\bsg)={\tt {\color{red}22222}00000}$ and $f_6(\bsg)={\tt {\color{red}222222}0000}$ are $\epsilon$-balanced. Hence, $t=4,5,6$ are $\epsilon$-balanced indices of $\bsg$. In general, there might be more than one $\epsilon$-balanced index.
\end{example}

The following result follows from Theorem~\ref{epsilon-balanced}.

\begin{corollary}\label{epsilon-balancedDNA}
Let $n$ be even, $\epsilon>0$. The set ${\rm S}_{\epsilon,n}$ is defined as in \eqref{set}. For any sequence $\bsg \in \Sigma_4^{n}$, there exists an index $t$ in the set ${\rm S}_{\epsilon,n}$, such that it is the $\epsilon$-balanced index of $\bsg$.
\end{corollary}

To encode a $\epsilon$-balanced sequence $\bsg$, we first find the smallest $\epsilon$-balanced index $t$ of $\bsg$, and then flip the first $t$ symbols of $\bsg$ according to the rule $f$. To represent the index, we also append a short balanced suffix to the end of codeword, and so, a lookup table of size $|{\rm S}_{\epsilon,n}|$ is required and the redundancy is $\ceil{\log \left( \floor{1/2\epsilon}+1 \right)}$. The following result is trivial.

\begin{lemma}\label{stillbalanced}
Let $n,m$ be even. Assume that $\bsg \in \Sigma_4^{n}$ is $\epsilon$-balanced and $\bz \in \Sigma_4^{m}$ is balanced. The concatenation sequence $\bsg \bz$ is also $\epsilon$-balanced.
\end{lemma}

\begin{example}\label{example}
Let $n=200, \epsilon=0.1$, i.e. we want the $\tG\tC$-content is within $[0.4,0.6]$, and the set ${\rm S}_{\epsilon,n}=\{0,40,80,120, 160, 200\}$ is of size six. We construct the one-to-one correspondence between the index and a short balanced suffix of length 2 as follows: $0 \to 02, 40 \to 03, 80 \to 12, 120 \to 13, 160 \to 20, 200 \to 30$. Assume that $\bsg \in \Sigma_4^{200}$ and the $\epsilon$-balanced index $t$ of $\bsg$ is $t=40$. The encoder flips the first 40 symbols in $\bsg$ to obtain $\bsg'$ that is $\epsilon$-balanced, and then append $03$ to the end of $\bsg'$. The encoder uses only two redundant symbols for $\epsilon=0.1$.
\end{example}

We now show that the suffix can be encoded and decoded in linear time without the use of a lookup table. In addition, in order to construct an $(\epsilon,\ell)$-constrained code, we encode the suffix in such a way that it is also $\ell$-runlength limited. The details are as follows.
\vspace{0.05in}

\noindent{\bf Index Encoder}. Let $n$ be even, $\epsilon,\ell>0$. The set ${\rm S}_{\epsilon,n}$ is defined as in \eqref{set}. Set $k \triangleq \ceil{\log_4 \left( \floor{1/2\epsilon}+1 \right)}$.

{\sc Input}: $t$, $t\in  {\rm S}_{\epsilon,n}, 0\leq t\leq n-1$\\
{\sc Output}: $\bp \triangleq \indexenc(t)$\\[-2mm] 

\begin{enumerate}[(I)]
\item Let $\tau_1\tau_2\cdots\tau_k$ be the quaternary representation of $t$ in ${\rm S}_{\epsilon,n}$.
\item Interleave the representation with the alternating length-$k$ sequence $f(\tau_1)f(\tau_2)\cdots f(\tau_k)$ to obtain $\bp$ of length $2k$. In other words, set $\bp=\tau_1f(\tau_1) \tau_2 f(\tau_2) \cdots \tau_k f(\tau_k)$.
\end{enumerate}

The corresponding \GC-content Encoder and Decoder are described as follows.
\vspace{0.05in}

\noindent{\bf $\epsilon$\GC-Encoder D}. Given $n,\epsilon>0$, set $k=\ceil{\log_4 \left( \floor{1/2\epsilon}+1 \right)}$ and $m=2n-4k$. Set ${\rm S}_{\epsilon,n-2k}$ be the set of indices as constructed in \eqref{set} and we construct a one-to-one correspondence between the indices in ${\rm S}_{\epsilon,n-2k}$ and $k$ bits sequences.

{\sc Input}: $\bx\in \{0,1\}^m$\\
{\sc Output}: $\bsg = \enc_{\epsilon\tG\tC}^{D}(\bx)$\\[-3mm] 

\begin{enumerate}[(I)]
\item Set $\bsg'=\Psi^{-1}(\bx) \in \Sigma_4^{n-2k}$
\item Search for the first $t$ in ${\rm S}_{\epsilon,n-2k}$, such that $t$ is the $\epsilon$-balanced index of $\bsg'$.
\item Obtain $\bsg''$ by flipping the first $t$ symbols in $\bsg'$.
\item Use Index Encoder to obtain $\bp$ representing index $t$ of length $2k$.
\item Finally, we set $\bsg\triangleq \bsg'' \bp$.
\end{enumerate}
\vspace{0.05in}


\noindent{\bf $\epsilon$\GC-Decoder D}. 

{\sc Input}: $\bsg\in \Sigma_4^n$, $\bsg$ is $\epsilon$-balanced\\
{\sc Output}: $\bx\triangleq \dec_{\epsilon\tG\tC}^{D}(\bsg)\in\{0,1\}^m$\\[-3mm] 

\begin{enumerate}[(I)]
\item Set $\bp$ be the suffix of length $2k$ in $\bsg$, and $\bsg'$ be the prefix of length $n-2k$.
\item Let $\bz$ be the sequence of odd indices in $\bp$, which is the $k$ bits sequence representing index $t$ in the set ${\rm S}_{\epsilon,n-2k}$.
\item Flip the first $t$ symbols in $\bsg'$ according to the flipping rule $f$ to obtain $\bsg''$.
\item Finally, output $\bx=\Psi(\bsg'')$
\end{enumerate}

\begin{remark}
The advantage of Encoder C is low redundancy, however, it is hard to combine with an RLL Encoder to construct an $(\epsilon,\ell)$-constrained encoder. In the next section, we present an efficient $(\epsilon,\ell)$-constrained encoder using the construction of Encoder D and the two RLL Encoders presented in Section~\ref{sec:RLL}.
\end{remark}



\section{Efficient $(\epsilon,\ell)$-Constrained Codes}
\label{sec:constrained}
In this section, we present an $(\epsilon,\ell)$-constrained encoder that translates binary data to DNA strands that are $\ell$-runlength limited and $\epsilon$-balanced for arbitrary $\epsilon,\ell>0$.
Prior to this work, literature results mostly focused on specific values of $\epsilon$ and $\ell$ \cite{wentu2018, dube2019}. For example, Song \et{}\cite{wentu2018} used concatenation technique to design RLL encoder for $\ell=3$, and their simulated results showed that the ${\tt G}{\tt C}$-content of all codewords is between 0.4 and 0.6, i.e. $\epsilon=0.1$, and for $n=200$, the rate of the encoder is 1.9 (bits/nt). In this section, we provide a more efficient coding scheme such that the output codewords are $\ell$-runlength limited and $\epsilon$-balanced. 


\begin{example}
Consider $n=10, \epsilon=0.1, \ell=3$. Let $\bsg=0002111011$.
Observe that even though $\bsg$ is $\ell$-runlength limited, it is not $\epsilon$-balanced. We then get $f_3(\bsg)={\tt {\color{red}222}2111011}$, is $\epsilon$-balanced. However, $f_3(\bsg)$ is not $\ell$-runlength limited.
\end{example}

The above example also illustrates that the sequence $f_t(\bsg)$ may not be $\ell$-runlength limited given that $\bsg$ is $\ell$-runlength limited. Nevertheless, we observe that the prefix and suffix of $f_t(\bsg)$ remain $\ell$-runlength limited. For brevity, given a sequence $\bsg \in \Sigma_4^n$, we use ${\rm P}_i(\bsg)$ and ${\rm S}_i(\bsg)$ to denote the prefix and suffix of $\bsg$ of length $i$, respectively.

\begin{lemma}\label{flip-rll}
Let $0\le t\le n$.
If a sequence $\bsg$ is $\ell$-runlength limited and $\bsg'=f_t(\bsg)$, then
${\rm P}_t(\bsg')$ and ${\rm S}_{n-t}(\bsg')$ are both $\ell$-runlength limited.
\end{lemma}
To ensure that the obtained sequence remains $\ell$-runlength limited, we simply add one redundant symbol before concatenating ${\rm P}_t(\bsg')$ and ${\rm S}_{n-t}(\bsg')$. 
\begin{corollary}[Concatenate two $\ell$-runlength limited sequences]\label{findgamma}
Let $\bsg,\bsg'$ be $\ell$-runlength limited. Suppose that the last symbol of $\bsg$ is $\alpha$ and the first symbol of $\bsg'$ is $\beta$. Let $\gamma \in \Sigma_4\setminus \{\alpha,\beta\}$, then $\bsg''= \bsg \gamma \bsg' $ is $\ell$-runlength limited.
\end{corollary}

We illustrate the construction of $(\epsilon,\ell)$-constrained encoder through the following example. 
\begin{example}[Example~\ref{example} continued]
Suppose $n=200, \epsilon=0.1$, and $\ell=3$. We show that there exists an efficient $(\epsilon,\ell)$-constrained encoder with at most $8$ redundant symbols. From the data sequence $\bsg\in\Sigma_4^{192}$, we use RLL Encoder A to obtain $\bsg_1=\enc_{\rm RLL}^A(\bsg)$. This step requires two redundant symbols and hence, $\bsg_1\in\Sigma_4^{194}$ is $\ell$-runlength limited.  We now search for the $\epsilon$-balanced index $t$ of $\bsg_1$ in the set ${\rm S}_{0.1,194}$ of size six, i.e $\bsg_2=f_t(\bsg_1)$ is $\epsilon$-balanced. Such index can be represented by a pointer $\bp$ of size two (similar to Example~\ref{example}). We follow Corollary~\ref{findgamma} to find $\gamma,\gamma'$ such that $\bsg_2= {\rm P}_t(\bsg_1) \gamma {\rm S}_{n-t}(\bsg_1) \gamma' \bp \in \Sigma_4^{198}$ be $\ell$-runlength limited. To ensure that the final output is $\epsilon$-balanced, recall that, ${\rm P}_t(\bsg_1) {\rm S}_{n-t}(\bsg_1) \bp$ is $\epsilon$-balanced, we then output $\bsg_3=\bsg_2 f(\gamma') f(\gamma)$. It is easy to verify that $\bsg_3$ is $\ell$-runlength limited and $\epsilon$-balanced. Thus, the encoder uses 8 redundant symbols to encode codewords of length 200, and hence, the rate is 1.92 (bits/nt).
\end{example}

We now show that the representation $\bp$ of the $\epsilon$-balanced index can be encoded/decoded in linear time without using a lookup table. Suppose we want to encode codewords in $\Sigma_4^n$ where $n$ is even. Set $k \triangleq \ceil{\log_4 \left( \floor{1/2\epsilon}+1 \right)}$, and $N=n-2k-4$.
Let $r_{\rm RLL}$ denote the number of redundant symbols used by the RLL Encoder ($\enc_{\rm RLL}^A$ or $\enc_{\rm RLL}^B$) to encode the $\ell$-runlength limited codewords in $\Sigma_4^{N}$. We summarize our proposed $(\epsilon,\ell)$-constrained encoder as follows. 
\vspace{0.05in}

\noindent{\bf $(\epsilon,\ell)$-Constrained Encoder}. Given $n,\epsilon,\ell$, $n$ even and $\ell\ge3$. Set $m=2n-2(r_{\rm RLL}+2k+4)$. Set ${\rm S}_{\epsilon,N}$ be the set of indices as defined by \eqref{set} and we construct a one-to-one correspondence between the indices in ${\rm S}_{N}$ and $k$ bits sequences.

{\sc Input}: $\bx \in \{0,1\}^m$\\
{\sc Output}: $\bsg \triangleq \enc_{(\epsilon,\ell)}(\bx) \in \Sigma_4^n$\\[-2mm]

\begin{enumerate}[(I)]
\item Set $\bsg_1=\Psi^{-1}(\bx) \in \Sigma_4^{n-r_{\rm RLL}-2k-4}$
\item Use RLL Encoder to obtain $\bsg_2=\enc_{\rm RLL}(\bsg_1)$, where $\bsg_2\in \Sigma_4^{N}$ is $\ell$-runlength limited
\item Search for the first $\epsilon$-balanced index $t$ of $\bsg_2$ in ${\rm S}_{\epsilon,N}$
\item Flip the first $t$ symbols in $\bsg_2$ to obtain $\bsg_3=f_t(\bsg_2)$
\item Let $\tau_1\tau_2\cdots\tau_k$ be the quaternary representation of $t$ in ${\rm S}_{\epsilon,N}$. Set $\bp=\tau_1f(\tau_1) \tau_2 f(\tau_2) \cdots \tau_k f(\tau_k)$
\item Use Corollary~\ref{findgamma} to find $\gamma$ and $\gamma'$ such that $\bsg_4={\rm P}_t(\bsg_3) \gamma {\rm S}_{N-t}(\bsg_3) \gamma' \bp$ is $\ell$-runlength limited
\item Output $\bsg=\bsg_4 f(\gamma) f(\gamma')$. Note that $\bsg\in\Sigma_4^n$
\end{enumerate}

\begin{theorem}\label{constrained proof}
The $(\epsilon,\ell)$-Constrained Encoder is correct. In other words, $\enc_{(\epsilon, \ell)}(\bx)$ is $\epsilon$-balanced and $\ell$-runlength limited for all $\bx \in \{0,1\}^m$. The redundancy of the encoder is $r_{\rm RLL}+2k+4$.
\end{theorem}
\begin{proof}
Let $\bsg=\enc_{(\epsilon, \ell)}(\bx)$. We first show that $\bsg$ is $\ell$-runlength limited. According to Corollary~\ref{findgamma}, $\bsg_4$ is $\ell$-runlength limited. Since two consecutive symbols in $\bp$ are distinct, the concatenation $\bp f(\gamma) f(\gamma')$ is $\ell$-runlength limited for all $\ell\ge3$. Therefore, $\bsg$ is $\ell$-runlength limited.

We now show that $\bsg$ is $\epsilon$-balanced. Since $\bsg_3$ is $\epsilon$-balanced, $\bp$ balanced, $\gamma f(\gamma), \gamma' f(\gamma')$ is balanced, we have $\bsg$ is $\epsilon$-balanced (according to Lemma~\ref{stillbalanced}).
\end{proof}

\begin{remark}\label{compare} The construction can be easily extended for $\ell\in\{1,2\}$. For arbitrary $\epsilon>0$, $k=\ceil{\log_4 \left( \floor{1/2\epsilon}+1 \right)}=O(1)$, is a constant. Therefore, the rate of this encoder approaches the rate of the RLL Encoder. If we use the RLL Encoder based on enumeration ($\enc_{\rm RLL}^A$) then the rate of the $(\epsilon,\ell)$-constrained encoder approaches the capacity for sufficient large $n$. However, this encoder A runs in $\Theta(n^2)$. For DNA storage with $\ell\in\{4,5,6\}$, we can use the linear time $\enc_{\rm RLL}^B$ to achieve as good rate as $\enc_{\rm RLL}^A$ (refer to Remark 9). 
\end{remark}

For completeness, we describe the corresponding $(\epsilon,\ell)$-constrained decoder as follows.

\noindent{\bf $(\epsilon,\ell)$-Constrained Decoder}. 

{\sc Input}: $\bsg \in \Sigma_4^n$, $\bsg$ is $\epsilon$-balanced and $\ell$-runlength limited\\
{\sc Output}: $\bx \triangleq \dec_{(\epsilon,\ell)}(\bsg) \in \{0,1\}^m$\\[-2mm]

\begin{enumerate}[(I)]
\item Set $\bp$ be the suffix of length $2k+2$ and $\bsg_1$ be the prefix of length $n-2k-3$
\item Remove the the last two symbols in $\bp$
\item Let $\bz$ be the sequence of odd indices in $\bp$, which is the $k$ bits sequence representing index $t$ in ${\rm S}_{\epsilon,N}$
\item Flip the first $t$ symbols in $\bsg_1$ according to the flipping rule $f$ to obtain $\bsg_2$
\item Remove the $(t+1)$th symbol in $\bsg_2$
\item Use RLL Decoder to obtain $\bsg_3=\dec_{\rm RLL}(\bsg_2)$
\item Output $\bx=\Psi(\bsg_3)$
\end{enumerate}


The efficiency of our designed $(\epsilon,\ell)$-constrained encoder are summarized in Table~\ref{capacity}. 
As can be seen, when the codeword length increases, the rate of our proposed encoder is only a few percent below capacity. 

\begin{table}[h]
\centering
\begin{tabular}{|c|| c | c| c|}
 \hline
 Codeword length $n$ & Capacity ${\bf C}$ & Rate of encoder ${\bf r}$ & $ \eta={\bf r}/{\bf C}$ (\%) \\ \hline
 $100$ & 1.99542 & 1.81000 & $90.707\%$\\
  $200$ &  1.99578& 1.92000 & $96.203\%$\\
   $300$ &  1.99577&  1.94000 &  $97.206\%$\\

 \hline
 \end{tabular}
  \caption{Rate of the designed constrained encoder for $\epsilon=0.1$ and $\ell=4$}.
\label{capacity}
\end{table}



\section{Efficient $(\epsilon,\ell;\B)$-Error-Control Codes}
\label{sec:error}
We now construct $(\epsilon,\ell;\B)$-error-control codes to correct the most common error in DNA data storage such as a single deletion, insertion, or substitution error. This also helps to reduce the error propagation of the constrained decoders proposed earlier. Crucial to our construction is the binary {\em Varshamov-Tenengolts (VT) codes} defined by Levenshtein \cite{le1965} and the $q$-ary VT codes defined by Tenengolts \cite{te1984}.

\subsection{Codes Correcting a Single Indel/Edit}
\begin{definition} The {\em binary VT syndrome} of a binary sequence $\bx\in\{0,1\}^n$ is defined to be
${\rm Syn}(\bx)=\sum_{i=1}^n i x_i$.
\end{definition}

For $a \in  \bbZ_{n+1}$, the Varshamov-Tenengolts code ${\rm VT}_a(n)$ is defined as follows.
\begin{equation}\label{VTcodes}
{\rm VT}_a(n)=\left\{\bx\in \{0,1\}^n: {\rm Syn}(\bx) = a \ppmod{n+1}\right\}.
\end{equation}

For $a \in  \bbZ_{n+1}$, the code ${\rm VT}_a(n)$ can correct a single indel and Levenshtein later provided a linear-time decoding algorithm \cite{le1965}.
To also correct a substitution, Levenshtein \cite{le1965} constructed the following code
\begin{equation}\label{VTsub}
{\rm L}_a(n)=\left\{\bx\in \{0,1\}^n: {\rm Syn}(\bx) = a \ppmod{2n}\right\},
\end{equation}
and provided a decoder that corrects a single edit.

\begin{theorem}[Levenshtein \cite{le1965}]\label{thm:lev}
Let ${\rm L}_a(n)$ be as defined in \eqref{VTsub}.
There exists a linear-time decoding algorithm $\dec^{\rm L}_a:\{0,1\}^{n*}\to {\rm L}_a(n)$ such that the following holds.
If $\bc\in {\rm L}_a(n)$ and $\by\in\Bedit(\bc)$,
then $\dec^{\rm L}_a(\by)=\bc$.
\end{theorem}

In 1984, Tenengolts \cite{te1984} generalized the binary VT codes to nonbinary ones.
Tenengolts defined the {\em signature} of a $q$-ary vector $\bx$ of length $n$ to be
the binary vector $\pi(\bx)$ of length $n-1$,
where $\pi(x)_i=1$ if $x_{i+1}\geq x_i$, and $0$ otherwise, for $i\in[n-1]$.
For $a\in \bbZ_n$ and $b\in \bbZ_q$, set
\begin{align*}
\label{qaryVT}
    {\rm T}_{a,b}({n;q}) \triangleq \big\{ &\bx \in \bbZ_q^n : \pi(\bx)\in{\rm VT}_a(n-1) \text{ and } \\
    &\sum_{i=1}^n x_i = b\ppmod{q} \big\}.
\end{align*}

Then Tenengolts showed that $T_{a,b}(n;q)$ corrects a single indel and
there exists $a$ and $b$ such that the size of ${\rm T}_{a,b}({n;q})$ is at least $q^n/(qn)$.
These codes are known to be asymptotically optimal. In the same paper, Tenengolts also provided a systematic $q$-ary single-indel-encoder
with redundancy $\log n +C_q$, where $n$ is the length of a codeword and $C_q$ is independent of $n$.

\begin{theorem}[Tenengolts \cite{te1984}]\label{thm:te}
 There exists a linear-time decoding algorithm $\dec^{\rm T}_{(a,b)}:\{0,1\}^{n*}\to {\rm T}_{a,b}({n;q})$ such that the following holds.
If $\bc\in {\rm T}_{a,b}({n;q})$ and $\by\in\Bindel(\bc)$,
then $\dec^{\rm T}_{(a,b)}(\by)=\bc$.
\end{theorem}

Recently, Chee \et{}\cite{tt2019} presented linear-time encoders for GC-balanced codewords that are capable of correcting single edit with $3 \log n + 2$ bits of redundancy. In the following, we use the idea of VT codes to modify the $(\epsilon,\ell)$-constrained code so that the codebook is capable of correcting either a single indel or a single edit.

For $\bsg\in\Sigma_4^n$, recall the definition of $\bU_{\bsg}, \bL_{\bsg} \in \{0,1\}^n$ and $\bx=\bU_{\bsg}|| \bL_{\bsg}=\Psi(\bsg)$ (refer to Section IV-III).

\begin{proposition}\label{prop:obs}
Let $\bsg\in\Sigma_4^n$. Then the following are true.
\begin{enumerate}[(a)]
\item $\bsg'\in \Bindel(\bsg)$ implies that $\bU_{\bsg'}\in \Bindel(\bU_\bsg)$ and $\bL_{\bsg'}\in \Bindel(\bL_\bsg)$.
\item $\bsg'\in \Bedit(\bsg)$ implies that $\bU_{\bsg'}\in \Bedit(\bU_\bsg)$ and $\bL_{\bsg'}\in \Bedit(\bL_\bsg)$.
\end{enumerate}
\end{proposition}

\begin{remark}\label{rem:observation}
The statement in Proposition~\ref{prop:obs} can be made stronger.
Suppose that there is an indel at position $i$ of $\bsg$.
Then there is exactly one indel at the same position $i$ in both upper and lower sequences of $\bsg$.
For example, consider $\bsg=020313$. We have $\bU_{\bsg}=010101$ and  $\bL_{\bsg}=000101$.
If the third symbol in $\bsg$, which is $0$, is deleted, we obtain $\bsg'=02313$ and hence,
$\bU'_{\bsg'}=01101$ and $\bL_{\bsg'}=00101$.
\end{remark}

The following construction is trivial.

\begin{corollary}\label{trivial-construction}
For $n>0, a\in\bbZ_{2n}, b\in\bbZ_{2n}$ , let $\C_{(a,b)}(n)$ be the set of all sequences $\bsg\in\Sigma_4^n$ such that $\bU_{\bsg}\in {\rm L}_a(n)$ and $\bL_{\bsg}\in{\rm L}_b(n)$. Then $\C_{(a,b)}(n)$ is capable of correcting a single edit error.
\end{corollary}

\subsection{Construction of $(\epsilon,\ell;\B^{indel})$-Error-Control Codes}

We follow Tenengolts's construction to encode DNA sequences that are capable of correcting a single indel. We simply append the information of the syndrome and the sum of symbols to the end of each codeword. In addition, we use the idea of the Index Encoder (refer to Section IV-D) to ensure the redundant part is balanced and $\ell$-runlength limited. The extra redundancy is $\log n + 4$. For simplicity, assume that $k'=\log n$ is integer and $k'$ is even.
\vspace{0.05in}

\noindent{\bf $(\epsilon,\ell;\B^{indel})$-Error-Control Encoder}. Let $n$ be even, $\epsilon,\ell>0$. Set $k \triangleq \ceil{\log_4 \left( \floor{1/2\epsilon}+1 \right)}$. 
Set $m=2n-2(r_{RLL}+2k+4)$, and $N=n-2k-4$. Set ${\rm S}_{\epsilon,n-2k-4}$ be the set of indices as defined by \eqref{set} and we construct a one-to-one correspondence between the indices in ${\rm S}_{\epsilon,n-2k-4}$ and $k$ bits sequences. Set $k'=\log n$.

{\sc Input}: $\bx \in \{0,1\}^m$\\
{\sc Output}: $\bsg \triangleq \enc_{(\epsilon,\ell;\Bindel)}(\bx) \in \C(\epsilon, \ell; \Bindel) \cap \Sigma_4^{n+\log n+4}$\\[-2mm]

\begin{enumerate}[(I)]
\item Use the $(\epsilon,\ell)$-constrained encoder to obtain $\bsg' = \enc_{(\epsilon,\ell)}(\bx) \in \Sigma_4^n$, where $\bsg'$ is $\epsilon$-balanced and $\ell$-runlength limited
\item Let $\alpha$ be the last symbol of $\bsg'$. Let $\beta$ be arbitrary symbol in $\Sigma_4\setminus\{\alpha, f(\alpha)\}$
\item Let $a = {\rm Syn}(\pi(\bsg')) \ppmod{n}$ and $b= \sum_{i=1}^n \bsg'_i \ppmod{4}$
\item Let $\tau_1\tau_2\cdots\tau_{k'/2}$ be the quaternary representation of $a$
\item Set $\bp=\beta f(\beta) \tau_1f(\tau_1) \tau_2 f(\tau_2) \cdots \tau_{k'/2} f(\tau_{k'/2}) b f(b)$
\item Output $\bsg=\bsg' \bp $
\end{enumerate}

\begin{theorem}
The $(\epsilon,\ell;\Bindel)$-error-control encoder is correct. In other words, $\enc_{(\epsilon, \ell;\Bindel)}(\bx)$ is $\epsilon$-balanced, $\ell$-runlength limited, and capable of correcting a single indel for all $\bx \in \{0,1\}^m$.
\end{theorem}
\begin{proof}
Let $\bsg=\enc_{(\epsilon, \ell;\Bindel)}(\bx)$. It is easy to show that $\bsg$ is $\epsilon$-balanced and $\ell$-runlength limited (refer to the proof of Theorem~\ref{constrained proof}). It remains to show that $\bsg$ can correct a single indel. To do so, we provide an efficient decoding algorithm. Suppose that there is a deletion (or insertion) in the received sequence $\bsg'$ (this can be observed based on the length of the received sequence). Without loss of generality, assume that the error is a deletion. The decoder proceeds as follows.

\noindent{\bf Localizing the deletion.} Let $\bp'$ be the suffix of length $k'+4$ of $\bsg'$. Assume that $\bp'=p'_1p'_2 \cdots p'_{k'+4}$.
\begin{itemize}
\item If $p'_2=f(p'_1)$ then we conclude that there is no deletion in $\bp$ and therefore, $\bp'\equiv \bp$.
\item If $p'_2\neq f(p'_1)$ then we conclude that there is a deletion in $\bp$.
\end{itemize}

\noindent{\bf Recovering $\bsg$.}
\begin{itemize}
\item If there is no deletion in $\bp$, i.e. $\bp'\equiv \bp$, let $\bsg''$ be the sequence obtained by removing the suffix $\bp$ from $\bsg'$. Note that ${\rm Syn}(\bsg'')$ and the sum of symbols in $\bsg''$ are known from $\bp$. We then set $\by=\dec^{\rm T}_{(a,b)}(\bsg'')$, and use the $(\epsilon,\ell)$-constrained encoder to obtain $\bx=\dec_{(\epsilon,\ell)}(\by)$.
\item If there is a deletion in $\bp$, we do not need to do error correction here, and remove the suffix of length $k'+3$ from $\bsg'$. We then use the $(\epsilon,\ell)$-constrained encoder to obtain $\bx=\dec_{(\epsilon,\ell)}(\bsg')$.
\end{itemize}
In conclusion, $\enc_{(\epsilon, \ell;\Bindel)}(\bx)$ is $\epsilon$-balanced, $\ell$-runlength limited, and can correct a single indel for all $\bx \in \{0,1\}^m$.
\end{proof}

\begin{corollary}
Let $M=n+\log n+4$. There exists a linear-time decoding algorithm $\dec_{\rm indel}:\Sigma_4^{M*}\to \C(\epsilon, \ell; \Bindel) \cap \Sigma_4^M$ such that the following holds.
If $\bsg=\enc_{(\epsilon,\ell;\Bindel)}(\bx)$ and $\bsg'\in\Bindel(\bsg)$,
then $\dec_{\rm indel}(\bsg')=\bsg$.
\end{corollary}

For completeness, we describe the corresponding $(\epsilon,\ell;\Bindel)$-error-control decoder as follows.

\noindent{\bf $(\epsilon,\ell;\B^{indel})$-Error-Control Decoder.} 

{\sc Input}: $\bsg' \in \Sigma_4^{(n+k'+4)*}$\\
{\sc Output}: $\bx \triangleq \dec_{(\epsilon,\ell;\Bindel)}(\bsg') \in \{0,1\}^m$\\[-2mm]

\begin{enumerate}[(I)]
\item Let $\bsg=\dec_{\rm indel}(\bsg') \in \Sigma_4^{n+k'+4}$
\item Use $(\epsilon,\ell)$-constrained decoder to obtain $\bx = \dec_{(\epsilon,\ell)}(\bsg)\in\{0,1\}^m$
\item Output $\bx$
\end{enumerate}


\subsection{Construction of $(\epsilon,\ell;\B^{edit})$-Error-Control Codes}
We follow the construction in Corollary~\ref{trivial-construction} to encode DNA sequences that are capable of correcting a single edit. We simply append the information of the syndrome of $\bU_{\bsg}$ and $\bL_{\bsg}$ to the end of each codeword. In addition, we also use the idea of the Index Encoder (refer to Section IV-D) to ensure the redundant part is balanced and $\ell$-runlength limited. The extra redundancy is $2\log n + 4$. For simplicity, assume that $k'=\log n$ is integer and $k'$ is even.
\vspace{0.05in}

\noindent{\bf $(\epsilon,\ell;\B^{edit})$-Error-Control Encoder}. Let $n$ be even, $\epsilon,\ell>0$. Set $k \triangleq \ceil{\log_4 \left( \floor{1/2\epsilon}+1 \right)}$. 
Set $m=2n-2(r_{RLL}+2k+4)$, and $N=n-2k-4$. Set ${\rm S}_{\epsilon,n-2k-4}$ be the set of indices as defined by \eqref{set} and we construct a one-to-one correspondence between the indices in ${\rm S}_{\epsilon,n-2k-4}$ and $k$ bits sequences. Set $k'=\log n$.

{\sc Input}: $\bx \in \{0,1\}^m$\\
{\sc Output}: $\bsg \triangleq \enc_{(\epsilon,\ell;\Bedit)}(\bx) \in \C(\epsilon, \ell; \Bedit) \cap \Sigma_4^{n+2\log n+4}$\\[-2mm]

\begin{enumerate}[(I)]
\item Use the $(\epsilon,\ell)$-constrained encoder to obtain $\bsg' = \enc_{(\epsilon,\ell)}(\bx) \in \Sigma_4^n$, where $\bsg'$ is $\epsilon$-balanced and $\ell$-runlength limited
\item Let $\alpha$ be the last symbol of $\bsg'$. Let $\beta$ be arbitrary symbol in $\Sigma_4\setminus\{\alpha, f(\alpha)\}$
\item Let $a = {\rm Syn}(\bU_{\bsg'})) \ppmod{n+1}$ and $b = {\rm Syn}(\bL_{\bsg'})) \ppmod{n+1}$, $c= \sum_{i=1}^n \bsg'_i \ppmod{4}$
\item Let ${\color{red}{\tau_1\tau_2\cdots\tau_{k'/2}}}$ be the quaternary representation of $a$, and ${\color{blue}{\nu_1\nu_2\cdots\nu_{k'/2}}}$ be the quaternary representation of $b$
\item Set $\bp=\beta f(\beta) {\color{red}{\tau_1f(\tau_1) \tau_2 f(\tau_2) \cdots \tau_{k'/2} f(\tau_{k'/2})}}$  ${\color{blue}{\nu_1f(\nu_1) \nu_2 f(\nu_2) \cdots \nu_{k'/2} f(\nu_{k'/2})}} c f(c)$
\item Output $\bsg=\bsg' \bp $
\end{enumerate}

\begin{theorem}
The $(\epsilon,\ell;\Bedit)$-error-control encoder is correct. In other words, $\enc_{(\epsilon, \ell;\Bedit)}(\bx)$ is $\epsilon$-balanced, $\ell$-runlength limited, and capable of correcting a single edit for all $\bx \in \{0,1\}^m$.
\end{theorem}
\begin{proof}
Let $\bsg=\enc_{(\epsilon, \ell;\Bedit)}(\bx)$. It is easy to show that $\bsg$ is $\epsilon$-balanced and $\ell$-runlength limited (refer to the proof of Theorem~\ref{constrained proof}). It remains to show that $\bsg$ can correct a single edit. To do so, we provide an efficient decoding algorithm. Suppose the received sequence is $\bsg'$. The idea is to recover the first $n$ symbols in $\bsg$ and then use the $(\epsilon,\ell)$-constrained decoder to recover the information sequence $\bx$. First, the decoder decides whether a deletion, insertion or substitution has occurred. Note that this information can be recovered by simply observing the length of the received sequence. The decoding operates as follows.

\begin{enumerate}[(i)]
\item If the length of $\bsg'$ is exactly $n+2\log n+4$, we conclude that at most a single substitution has occurred.
\begin{itemize}
\item Let $\bp'$ be the suffix of length $2\log n+4$ of $\bsg'$, and $\bp'=p'_1 p'_2 \cdots p'_{2k'+4}$.
\item Let $\bsg''$ be the prefix of length $n$ of $\bsg'$. The decoder computes ${\rm Syn}(\bU_{\bsg''})$ and ${\rm Syn}(\bL_{\bsg''}) \ppmod{n+1}$.
\item Let $a'$ be the integer number whose quaternary representation is $p'_3 p'_5 \cdots p'_{k'+1}$, $b'$ be the integer number whose quaternary representation is $p'_{k'+3} p'_{k'+5} \cdots p'_{2k'+1}$ and $c'=p'_{2k'+3}$.
\item If $c'$ is the sum of symbols in $\bsg''$, then there is no error in $\bsg''$. The decoder proceeds to obtain $\bx=\dec_{(\epsilon,\ell)}(\bsg'')$. Otherwise, if $a'={\rm Syn}(\bU_{\bsg''})$ and $b'={\rm Syn}(\bU_{\bsg''})$ then there is no error in $\bsg''$, the decoder proceeds to obtain $\bx=\dec_{(\epsilon,\ell)}(\bsg'')$. On the other hand, if either one statement is false, there is an error in $\bsg''$. The decoder sets $\by=\dec^{\rm L}_{a'}(\bU_{\bsg''})$ and $\bz=\dec^{\rm L}_{b'}(\bL_{\bsg''})$. Finally, $\bsg=\Psi(\by||\bz)$ and the decoder returns $\bx=\dec_{(\epsilon,\ell)}(\bsg)$.
\end{itemize}
\item If the length of $\bsg'$ is exactly $n+2\log n+3$, we conclude that a single deletion has occurred (the case of single insertion can be done similarly). The decoder proceeds as follows.
\begin{itemize}
\item Let $\bp'$ be the suffix of length $2\log n+4$ of $\bsg'$, and $\bp'=p'_1 p'_2 \cdots p'_{2k'+4}$.
\item If $p'_2\neq f(p'_1)$, the decoder concludes that there is a deletion in $\bp$. The decoder removes the suffix of length $2k'+3$ from $\bsg'$, then use the $(\epsilon,\ell)$-constrained encoder to obtain $\bx=\dec_{(\epsilon,\ell)}(\bsg')$
\item If $p'_2=f(p'_1)$, the decoder concludes that there is no deletion in $\bp$ and therefore, $\bp'\equiv \bp$. Let $\bsg''$ be the sequence obtained by removing the suffix $\bp$ from $\bsg'$. Note that ${\rm Syn}(\bU_{\bsg''})$ and ${\rm Syn}(\bL_{\bsg''})$ are known from $\bp$. The decoder sets $\by=\dec^{\rm L}_{a}(\bU_{\bsg''})$ and $\bz=\dec^{\rm L}_{b}(\bL_{\bsg''})$. Finally, $\bsg=\Psi(\by||\bz)$ and the decoder returns $\bx=\dec_{(\epsilon,\ell)}(\bsg)$.
\end{itemize}
\end{enumerate}

In conclusion, $\enc_{(\epsilon, \ell;\Bedit)}(\bx)$ is $\epsilon$-balanced, $\ell$-runlength limited, and can correct a single edit for all $\bx \in \{0,1\}^m$.
\end{proof}

\begin{corollary}
Let $M=n+2\log n+4$. There exists a linear-time decoding algorithm $\dec_{\rm edit}:\Sigma_4^{M*}\to \C(\epsilon, \ell; \Bedit) \cap \Sigma_4^M$ such that the following holds.
If $\bsg=\enc_{(\epsilon,\ell;\Bedit)}(\bx)$ and $\bsg'\in\Bedit(\bsg)$,
then $\dec_{\rm edit}(\bsg')=\bsg$.
\end{corollary}

For completeness, we describe the corresponding $(\epsilon,\ell;\Bedit)$-error-control decoder as follows.
\vspace{0.05in}

\noindent{\bf $(\epsilon,\ell;\B^{edit})$-Error-Control Decoder.} 

{\sc Input}: $\bsg' \in \Sigma_4^{(n+2\log n+4)*}$\\
{\sc Output}: $\bx \triangleq \dec_{(\epsilon,\ell;\Bedit)}(\bsg') \in \{0,1\}^m$\\[-2mm]

\begin{enumerate}[(I)]
\item Let $\bsg=\dec_{\rm edit}(\bsg') \in \Sigma_4^{n+2\log n+4}$
\item Use $(\epsilon,\ell)$-constrained decoder to obtain $\bx = \dec_{(\epsilon,\ell)}(\bsg)\in\{0,1\}^m$
\item Output $\bx$
\end{enumerate}

\begin{remark} 
We use $r_{error}$ to denote the redundancy needed to correct single indel or edit error. When $\B=\Bindel$, $r_{error}=\log n+4$, and when $\B=\Bedit$, $r_{error}=2\log n+4$. Since $\frac{\log n}{n} \to 0$,  $r_{GC} = O(1)$, is a constant, the rate of this encoder approaches the rate of the RLL Encoder, and if we use RLL Encoder A then the rate of the $(\epsilon,\ell,\B)$-error-control encoder approaches the capacity for sufficient large $n$. 
\end{remark}

\section{Conclusion}
We have presented novel and efficient encoders that translate binary data into strands of nucleotides which satisfy the RLL constraint, the \GC-content constraint, and are capable of correcting a single edit and its variants. Our proposed codes achieve higher rates than previous results and approach capacity, have low encoding/decoding complexity and limited error propagation.

\begin{thebibliography}{9}


\bibitem{Yazdi.2017}
	S.~Yazdi, R.~Gabrys,  and O.~Milenkovic,
	``Portable and error-free DNA-based data storage'', \emph{Scientific Reports}, no.~5011, vol.~7, 2017.
	
\bibitem{church2012} G. M. Church, Y. Gao, and S. Kosuri, ``Next-generation digital information storage in DNA," {\em Science}, vol. 337, no. 6102, pp. 1628-1628, 2012.

\bibitem{goldman2013} N. Goldman, P. Bertone, S. Chen, C. Dessimoz, E. M. LeProust, B. Sipos, and E. Birney, ``Towards practical, high-capacity, low-maintenance information storage in synthesized DNA," {\em Nature}, vol. 494, no. 7435, pp. 77-80, 2013.


\bibitem{fountain2017} Y. Erlich and D. Zielinski, ``DNA fountain enables a robust and efficient storage architecture," {\em Science}, vol. 355, no. 6328, pp. 950-954, 2017.

\bibitem{Organick:2018}
L.~Organick, S.~Ang, Y.~J.~Chen, R.~Lopez, S.~Yekhanin, K.~Makarychev, M.~Racz, G.~Kamath, P.~Gopalan, B.~Nguyen, C.~Takahashi, S.~Newman, H.~Y.~Parker, C.~Rashtchian, K.~Stewart, G.~Gupta, R.~Carlson, J.~ Mulligan, D.~Carmean, G.~Seelig, L.~Ceze, and K.~Strauss,
``Random access in large-scale {DNA} data storage",
\emph{Nature Biotechnology},
vol.~36, no.~3, 242--248, 2018.

\bibitem{ross2013} M. G. Ross, C. Russ, M. Costello, A. Hollinger, N. J. Lennon, R. Hegarty, C. Nusbaum, and D. B. Jaffe, ``Characterizing and measuring bias in sequence data", {\em Genome Biology}, vol. 14, 2013.

\bibitem{exp1}
R.~Heckel, G.~Mikutis, and R. N. Grass, ``A Characterization of the DNA Data Storage Channel", {\em Scientific Reports}, Jul. 2019.

\bibitem{immink2018} K. A. S. Immink, and K. Cai, ``Design of Capacity-Approaching Constrained Codes for DNA-Based Data Storage Systems," {\em IEEE Communications Letters}, vol. 22, no. 2, pp. 224-227, 2018.

\bibitem{tt2019} K. Cai, Y. M. Chee, R. Gabrys, H. M. Kiah, and T. T. Nguyen, ``Optimal Codes Correcting a Single Indel / Edit for DNA-Based Data Storage", preprint, {\em arXiv}, arXiv:1910.06501, 2019.

\bibitem{demarau} R. Gabrys, E. Yaakobi, and O. Milenkovic, ``Codes in the Damerau Distance for Deletion and Adjacent Transposition Correction", {\em IEEE Trans. Inform. Theory}, Vol. 64, No. 4, 2018.

\bibitem{wentu2018} W. Song, K. Cai, M. Zhang, and C. Yuen, ``Codes with Run-Length and GC-Content Constraints for DNA-based Data Storage," {\em IEEE Communications Letters}, vol. 22 , no. 10, pp. 2004-2007, Oct. 2018.

\bibitem{dube2019} D. Dube, W. Song, and K. Cai, ``DNA Codes with Run-Length Limitation and Knuth-Like Balancing of the GC Contents", {\em Symposium on Information Theory and its Applications (SITA)}, Japan, Nov. 2019.


\bibitem{Yakovchuk2006} P. Yakovchuk, E. Protozanova, and M. D. Frank-Kamenetskii, ``Base-stacking and base-pairing contributions into thermal stability of the DNA double helix", {\em Nucl. Acids Res.}, vol. 34, no. 2, pp. 564-574, 2006.





\bibitem{knuth} D. E. Knuth, ``Efficient Balanced Codes", {\em IEEE Trans. Inform. Theory}, vol. IT-32, no. 1, pp. 51-53, Jan 1986.





\bibitem{W2010} A. J. de Lind van Wijngaarden and K. A. S. Immink, ``Construction of Maximum Run-Length Limited Codes Using Sequence Replacement Techniques," {\em IEEE Journal on Selected Areas of Communications}, vol. 28, pp. 200-207, 2010.

\bibitem{O2019} O. Elishco, R. Gabrys, M. Medard, and E. Yaakobi, ``Repeated-Free Codes", {\em Proc. IEEE Int. Symp. Inf. Theory (ISIT)}, Paris, France,  2019.

\bibitem{schoeny2017} C. Schoeny, A. Wachter-Zeh, R. Gabrys, and E. Yaakobi, ``Codes correcting a burst of deletions or insertions?, {\em IEEE Trans. Inform. Theory}, vol. 63, no. 4, pp. 1971-1985, 2017.

\bibitem{wijk1972} J. P. M. Schalkwijk, ``An algorithm for source coding," {\em IEEE Trans. Inf. Theory}, IT-18, pp. 395-399, 1972.

\bibitem{1988} N. Alon, E. E. Bergmann, D. Coppersmith, and A. M. Odlyzko, ``Balancing sets of vectors", {\em IEEE Trans. Inf. Theory}, vol. IT-34, no. 1, pp. 128-130, Jan. 1988.

\bibitem{imbalanced} V. Skachek and K. A. S. Immink, ``Constant Weight Codes: An Approach Based on Knuth's Balancing Method", {\em IEEE Journal on Selected Areas in Communications}, vol. 32, No. 5, May 2014.

\bibitem{bose1996} L. G. Tallini, R. M. Capocelli, and B. Bose, ``Design of some new balanced codes," {\em IEEE Trans. Inf. Theory}, vol. IT-42, pp. 790-802, May 1996.

\bibitem{le1965}
V. I. Levenshtein, ``Binary codes capable of correcting deletions, insertions and reversals", {\em Doklady Akademii Nauk SSSR},
vol. 163, no. 4, pp. 845-848, 1965.

\bibitem{te1984} G. Tenengolts, ``Nonbinary codes, correcting single deletion or insertion", {\em IEEE Trans. Inf. Theory}, vol. 30, no. 5, pp. 766-769, 1984.

\end{thebibliography}
\end{document}